\PassOptionsToPackage{unicode}{hyperref}
\PassOptionsToPackage{hyphens}{url}
\documentclass[11pt, ]{article}

\usepackage{mathtools}
\usepackage{amsmath}
\usepackage{amsthm}
\usepackage{amssymb}
\usepackage[italicdiff]{physics}

\usepackage{iftex}
\ifPDFTeX
  \usepackage[OT1,T1]{fontenc}
  \usepackage[utf8]{inputenc}
  \usepackage{textcomp} 
    \usepackage[p,osf,swashQ]{cochineal}
  \usepackage[cochineal,vvarbb]{newtxmath}
      \usepackage[scale=0.95]{biolinum}
    \usepackage[scale=0.95,varl]{inconsolata}
\else 
  \usepackage[scale=0.95,varl]{inconsolata}
  \usepackage{newpxtext}
  \usepackage{mathpazo}
    \usepackage[scale=0.95]{biolinum}
  \fi
\ifLuaTeX
  \usepackage{selnolig}  
\fi
\IfFileExists{microtype.sty}{
  \usepackage[]{microtype}
  \UseMicrotypeSet[protrusion]{basicmath} 
}{}

\allowdisplaybreaks
\usepackage[dvipsnames,svgnames,x11names]{xcolor}
\usepackage[lmargin=1in,rmargin=1in,tmargin=1in,bmargin=1in]{geometry}
\usepackage[format=plain,
  labelfont={bf,sf,small,singlespacing},
  textfont={sf,small,singlespacing},
  justification=justified,
  margin=0.25in]{caption}

\usepackage{sectsty}
\usepackage[compact]{titlesec}
\makeatletter
\newcommand\@shorttitle{}
\newcommand\shorttitle[1]{\renewcommand\@shorttitle{#1}}
\usepackage{fancyhdr}
\fancyheadoffset{0pt}
\makeatother
\renewenvironment{abstract}{ 
  \centerline
  {\large\sffamily\bfseries Abstract}\vspace{-1em}
  \begin{quote}\small
}{
  \end{quote}
}

\makeatletter
\newcommand{\assumplabel}[2]{%
   \protected@write \@auxout {}{\string\newlabel{#1}{{#2}{\thepage}{#2}{#1}{}}}%
   \hypertarget{#1}{#2}%
}
\makeatother

\providecommand{\tightlist}{%
  \setlength{\itemsep}{0pt}\setlength{\parskip}{0pt}}\usepackage{longtable,booktabs,array}
\usepackage{calc} 
\usepackage{etoolbox}
\makeatletter
\patchcmd\longtable{\par}{\if@noskipsec\mbox{}\fi\par}{}{}
\makeatother
\IfFileExists{footnotehyper.sty}{\usepackage{footnotehyper}}{\usepackage{footnote}}
\makesavenoteenv{longtable}
\usepackage{graphicx}
\makeatletter
\def\maxwidth{\ifdim\Gin@nat@width>\linewidth\linewidth\else\Gin@nat@width\fi}
\def\maxheight{\ifdim\Gin@nat@height>\textheight\textheight\else\Gin@nat@height\fi}
\makeatother
\setkeys{Gin}{width=\maxwidth,height=\maxheight,keepaspectratio}
\makeatletter
\def\fps@figure{htbp}
\makeatother

\usepackage{thmtools}
\usepackage{algorithm}

\theoremstyle{plain}

\newcommand{\R}{\ensuremath{\mathbb{R}}}

\DeclareMathOperator*{\argmax}{arg\,max}

\newcommand{\cvas}{\xrightarrow{\;\!a.s.\:\!}}

\newcommand{\M}{\mathcal{M}}
\makeatletter
\@ifpackageloaded{caption}{}{\usepackage{caption}}
\AtBeginDocument{%
\ifdefined\contentsname
  \renewcommand*\contentsname{Table of contents}
\else
  \newcommand\contentsname{Table of contents}
\fi
\ifdefined\listfigurename
  \renewcommand*\listfigurename{List of Figures}
\else
  \newcommand\listfigurename{List of Figures}
\fi
\ifdefined\listtablename
  \renewcommand*\listtablename{List of Tables}
\else
  \newcommand\listtablename{List of Tables}
\fi
\ifdefined\figurename
  \renewcommand*\figurename{Figure}
\else
  \newcommand\figurename{Figure}
\fi
\ifdefined\tablename
  \renewcommand*\tablename{Table}
\else
  \newcommand\tablename{Table}
\fi
}
\@ifpackageloaded{float}{}{\usepackage{float}}
\floatstyle{ruled}
\@ifundefined{c@chapter}{\newfloat{codelisting}{h}{lop}}{\newfloat{codelisting}{h}{lop}[chapter]}
\floatname{codelisting}{Listing}

\makeatother
\makeatletter
\@ifpackageloaded{caption}{}{\usepackage{caption}}
\@ifpackageloaded{subcaption}{}{\usepackage{subcaption}}
\makeatother

\usepackage[]{natbib}
\newenvironment{CSLReferences}[2]{ 
\bibliography{references.bib}
\clearpage
}{}

\usepackage{hyperref}
\usepackage{url}
\hypersetup{
  pdftitle={Finding Pareto Efficient Redistricting Plans with Short Bursts},
  pdfauthor={Cory McCartan},
  pdfkeywords={redistricting, optimization, Markov chain, Pareto
efficiency},
  colorlinks=true,
  linkcolor={black},
  filecolor={Maroon},
  citecolor={VioletRed4},
  urlcolor={DodgerBlue4},
  pdfcreator={LaTeX via pandoc}}

\title{\sffamily\bfseries\huge\parfillskip=0pt
\rightskip=0pt plus .5\textwidth
\leftskip=0pt plus .5\textwidth
\emergencystretch=.3\textwidth Finding Pareto Efficient Redistricting
Plans with Short Bursts}
\shorttitle{Finding Pareto Efficient Redistricting Plans with Short Bursts}
\author{\textbf{Cory McCartan}\footnote{
To whom correspondence should be addressed.
Email: \texttt{\href{mailto:cmccartan@g.harvard.edu}{cmccartan@g.harvard.edu}}.
Website: \url{https://corymccartan.com}.
Address:
1 Oxford St, Cambridge, MA 02138.
The author thanks Christopher Kenny and Tyler Simko for helpful comments
and suggestions.}
\\Department of Statistics%
\\Harvard University%
\vspace{2pt}
 }
\date{May 27, 2024}

\begin{document}
\allsectionsfont{\sffamily}

\maketitle

\begin{abstract}
Redistricting practitioners must balance many competing constraints and
criteria when drawing district boundaries. To aid in this process,
researchers have developed many methods for optimizing districting plans
according to one or more criteria. This research note extends a
recently-proposed single-criterion optimization method, \emph{short
bursts} \citep{cannon2023voting}, to handle the multi-criterion case,
and in doing so approximate the Pareto frontier for any set of
constraints. We study the empirical performance of the method in a
realistic setting and find it behaves as expected and is not very
sensitive to algorithmic parameters. The proposed approach, which is
implemented in open-source software, should allow researchers and
practitioners to better understand the tradeoffs inherent to the
redistricting process.
\end{abstract}

\textbf{\textit{Keywords}}\quad redistricting~\textbullet~optimization~\textbullet~Markov
chain~\textbullet~Pareto efficiency


\section{Introduction}\label{sec-intro}

Legislative districts in the U.S. must satisfy a wide variety of
constraints and criteria, which vary across states and municipalities,
in addition to several federal statutory and constitutional criteria
\citep{ncslcriteria}. Many recent advances in the quantitative analysis
of redistricting have involved introducing techniques for generating
sample districting plans which satisfy these constraints. These
techniques fall into two broad categories: sampling algorithms and
optimization algorithms.

Sampling algorithms aim to generate a random sample of districting plans
which meet a set of constraints. These algorithms include ad-hoc methods
\citep{cirincione2000, chen2013, magleby2018} as well as algorithms that
are designed to sample from a specific probability distribution
\citep{mattingly2014, wu2015, deford2019, carter2019, fifield2020automated, mccartan2020, cannon2022spanning};
most of which are based on Markov chain Monte Carlo techniques.

In contrast, optimization algorithms are designed to produce a
districting plan or plans which score well on predefined numerical
criteria. Informally, while sampling algorithms can give an idea of a
``typical'' districting plan, optimization algorithms can find more
extreme or even close-to-optimal districting plans along certain
dimensions. Many different optimization schemes have been proposed in
the literature
\citep{mehrotra1998, macmillan2001, bozkaya2003, altman2011bard, liu2016pear, rincon2013multiobjective, gurnee2021fairmandering, swamy2022multiobjective}.
One recent development by \citet{cannon2023voting} has been optimization
by \emph{short bursts}, which has shown particular promise in maximizing
minority vote share across many districts. The short burst algorithm
operates by running a Markov chain on the space of districting plans for
a small number of steps, then restarting it from the step which scored
highest according to a prespecified criterion.

One limitation of many of these optimization algorithms is that users
must specify a univariate scoring function. When multiple constraints
must be balanced against each other, practitioners often take a linear
combination of multiple scoring functions and carefully tune the weights
to achieve their desired goal. An alternative approach, which has been
previously proposed, is to use optimization or sampling algorithms to
find districting plans on the Pareto frontier for a set of constraints:
plans which cannot be improved on any one dimension without sacrificing
another dimension \citep{gerken2010getting, altman2018redistricting}.
The ability to identify a set of plans lying along or close to the
Pareto frontier can be a valuable tool for redistricting researchers and
practitioners alike in understanding the compatibility and tradeoffs
between various redistricting criteria.

While some optimization algorithms
\citep{rincon2013multiobjective, swamy2022multiobjective} are explicitly
designed to handle multiple constraints and approximate the Pareto
frontier, many, including optimization by short bursts, are not. This
research note provides a natural extension of the short burst algorithm
to optimize the entire Pareto frontier at once.

Section~\ref{sec-method} describes the problem formally and develops the
algorithmic extension, which is also implemented in open-source software
\citep{redist}. We examine the efficacy of the proposed algorithm in
studying congressional redistricting in Iowa in Section~\ref{sec-demo}.
We find that the proposed algorithm identifies an expanding Pareto
frontier as the number of bursts increases, that the algorithm's
performance is not sensitive to the burst size, and that multiple
independent runs of the algorithm for fixed parameters produce similar
results. Section~\ref{sec-conclude} concludes and discusses directions
for future work.

\section{Approximating the Pareto Frontier with Short
Bursts}\label{sec-method}

Let \(\Xi\) denote the (finite) collection of possible districting plans
in a state or city. In most algorithmic redistricting work, \(\Xi\) is
defined as the set of possible graph partitions of an adjacency graph
over geographic units such as precincts or Census tracts \citep[see,
e.g.,][]{mccartan2020}.

\subsection{Pareto ordering and
efficiency}\label{pareto-ordering-and-efficiency}

Suppose we have a scoring function \(f:\Xi\to\R^J\) that provides a
vector of \(J\) numerical scores for a given districting plan. We refer
a single element of the vector \(f_j\) as a districting criterion.
Without loss of generality we interpret larger values of each \(f_j\) to
be desirable for the practitioner. We can then define a strict partial
order on districting plans \(\xi\in\Xi\) as follows: \(\xi\prec\xi'\)
iff for all \(j\), \(f_j(\xi)\le f_j(\xi')\), and there exists a \(j\)
with \(f_j(\xi)<f_j(\xi')\). Note that when \(J=1\) this is in fact a
strict total order. This ordering is called the \emph{Pareto ordering};
when \(\xi\prec\xi'\) we say that \(\xi'\) \emph{Pareto dominates}
\(\xi\).

A plan is \emph{Pareto efficient} if there does not exist a distinct
\(\xi'\in\Xi\) that dominates it. The \emph{Pareto frontier} \(P(f)\) is
the set of all Pareto efficient plans. It is not possible to improve a
plan on the Pareto frontier (in the sense of increasing the value of a
criterion \(f_j\)) without decreasing the value of at least one
criterion.

\subsection{Pareto optimization by short
bursts}\label{pareto-optimization-by-short-bursts}

The Pareto frontier \(P(f)\) for a scoring function \(f:\Xi\to\R^J\) is
easily calculated from an enumeration of all possible districting plans:
simply discard plans in \(\Xi\) which are Pareto dominated by any other
plan. Unfortunately, in most realistic districting problems, \(\Xi\) is
too large to enumerate, and so analysts must resort to approximations.

Some researchers have generated an approximate Pareto frontier by
generating a large number of samples from \(\Xi\) using a redistricting
sampling algorithm targeting a ``neutral'' distribution \(\pi\) which
contains no information about \(f\), and then discarding samples which
are Pareto dominated by any other sampled plan
\citep{schutzman2020trade}. While this approach will consistently
generate the Pareto frontier as the number of samples goes to infinity,
it will not perform well in general in finite samples, because the
generation of the sample from \(\Xi\) is done without any knowledge of
the scoring function \(f\).

One could imagine instead generating samples from a target distribution
\(\pi_f\) which places more probability mass on plans that score well on
\(f\), and then discarding Pareto-dominated plans. For example,
\(\pi_f(\xi)=\exp(-\norm{f(\xi)})\) is one such distribution. As
\citet{cannon2023voting} show, however, a better approach yet might be
to generate the Pareto frontier over a series of ``short bursts.''

The original short bursts algorithm of \citet{cannon2023voting} is
described in Algorithm \ref{alg:sb}. Given an initial plan \(\xi_0\),
the short burst algorithm runs \(b\) steps of some Markov chain \(\M\)
on districting plans, then picks the best plan from the set of \(b+1\)
plans (including the initializing plan) to initialize the next
``burst.'' The chain \(\M\) can be any Markov chain on districting
plans, even one which does not satisfy detailed balance.
\citet{cannon2023voting} use the Markov chain of \citet{deford2019};
here, we use a substantially similar algorithm that is also
spanning-tree based and is implemented in the software of
\citet{redist}. The short burst algorithm is terminated when
\(f(\xi_0)\) reaches a predefined threshold, or when the total number of
bursts run hits a specified maximum.

\begin{algorithm}[ht]
\caption{Univariate optimization by short bursts}
\label{alg:sb}
\textit{Input}: initial plan $\xi_0$, univariate scoring function $f:\Xi\to\R$,
Markov chain $\M$, and burst size parameter $b$.

\textbf{Repeat} until termination:
\begin{enumerate}\tightlist
\item Run $\M$ for $b$ steps starting at $\xi_0$, producing plans $(\xi_1,\dots, \xi_b)$.
\item Set $k \leftarrow \argmax_{0\le i\le b} f(\xi_i)$.
\item Set $\xi_0 \leftarrow \xi_k$.
\end{enumerate}
\end{algorithm}

For a scoring function \(f:\Xi\to\R^J\), the Pareto frontier \(P(f)\)
can be approximated with Algorithm \ref{alg:sb} by running the algorithm
to completion on many different scoring functions \(f_{\vec w}\), where
\[
f_{\vec w}(\xi) = \sum_{j=1}^J w_j f_j(\xi),
\] for many different combinations of weights \(\vec w\) lying in a
\(J\)-simplex. However, this is computationally expensive, especially as
\(J\) grows.

Instead, we propose a simple modification to Algorithm \ref{alg:sb}.
Rather than selecting the best plan from the previous burst to
initialize the next burst, we will keep track of the set of sampled
plans which are not Pareto-dominated by any other observed plan, and
initializing each burst with a plan sampled from this Pareto efficient
set. This yields a generalized short burst algorithm for Pareto
optimization which is detailed in Algorithm \ref{alg:sbp}. The
univariate Algorithm \ref{alg:sb} is a special case of Algorithm
\ref{alg:sbp} with \(J=1\).

\begin{algorithm}[ht]
\caption{Pareto optimization by short bursts}
\label{alg:sbp}
\textit{Input}: set of initial plans $X$, scoring function $f:\Xi\to\R^J$,
Markov chain $\M$, and burst size parameter $b$.

\textbf{Repeat} until termination:
\begin{enumerate}\tightlist
\item Sample $\xi_0$ from $X$ uniformly at random.
\item Run $\M$ for $b$ steps starting at $\xi_0$, producing plans $(\xi_1,\dots, \xi_b)$.
\item Set $X\leftarrow X\cup \{\xi_1,\dots, \xi_b\}$.
\item Remove Pareto-dominated plans from $X$: set 
$X\leftarrow \{\xi\in X: \not\exists\ \xi'\in X \text{ with } \xi\prec\xi'\}$.
\end{enumerate}
\end{algorithm}

By initializing each burst with a sample from the yet-observed Pareto
frontier, Algorithm \ref{alg:sbp} is able to explore and improve the
entire Pareto frontier at any point. One could imagine generalizations
of the algorithm which place more weight on plans on the Pareto frontier
which the algorithm heuristically believes are easier to further
optimize. We leave such improvements to future work.

\subsection{Algorithm properties}\label{algorithm-properties}

Letting \(X_n\) denote the (random) output of Algorithm \ref{alg:sbp}
after \(n\) bursts, we can record several easily-verified properties of
the proposed algorithm. Proofs are deferred to the appendix.

\begin{restatable}{prop}{propconv}
Let $\M^b$ be the Markov chain obtained by taking $b$ steps from $\M$ at a time.
If the Markov chain $\M^b$ has a strictly positive transition probability between any pair of plans, then $X_n\cvas P(f)$ as $n\to\infty$.
\end{restatable}

Unfortunately, all Markov chains developed to date to sample districting
plans do not have strictly positive transition probability between all
pairs of plans. Weaker conditions such as the irreducibility of \(\M\)
are not enough for convergence, either, as the following proposition
records.

\begin{restatable}{prop}{propirred}
\label{prop:irred}
For any $\Xi$, burst size $b\le |\Xi|-2$, and scoring function $f:\Xi\to\R$ which takes at least $b+2$ distinct values on  $\Xi$, there exists an irreducible Markov chain $\M$ on $\Xi$ such that $X_n\not\cvas P(f)$ as $n\to\infty$ for some initializing set $A$.
\end{restatable}

While Algorithm \ref{alg:sbp} is therefore not guaranteed to produce the
Pareto frontier in all settings, even with infinite computing time,
given the promising results in \citet{cannon2023voting}, we expect it to
perform well in approximating the Pareto frontier in real-world
redistricting problems.

One qualitative difference in algorithm performance between the \(J=1\)
and \(J>1\) settings involves local minima. When \(J=1\) and the
algorithm finds itself in a local minimum, it is often unable to further
improve the scoring function. Multiple runs of the algorithm are
therefore often required to have confidence that any particular run has
avoided such a local minimum. While there is no guarantee of avoiding
local minima when \(J>1\), in practice such issues may be mitigated by
the fact that the algorithm randomly selects starting points along the
currently-estimated Pareto frontier for each burst. If one point along
the frontier lies in a local minimum, the algorithm can still improve
the overall frontier by starting bursts from other points. Of course,
the overall size of the Pareto frontier grows rapidly as \(J\) increases
(see appendix). As a result, each portion of the frontier may get less
overall ``attention'' per burst, leading to overall slower improvements
in any one region of the frontier.

\begin{figure}

\centering{

\includegraphics[width=0.75\textwidth,height=\textheight]{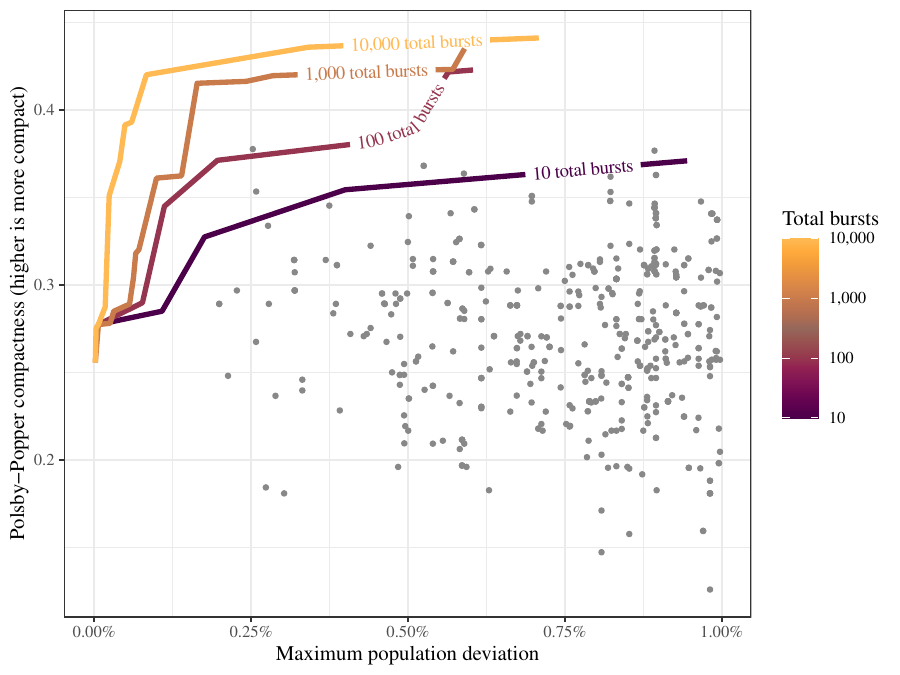}

}

\caption{\label{fig-demo}Pareto frontier estimated by the proposed
method across a range of total bursts (colored lines), plotted against
1,000 samples from an MCMC sampling algorithm using the same Markov
chain (grey points).}

\end{figure}%

\section{Empirical Demonstration}\label{sec-demo}

Here we apply the proposed Algorithm \ref{alg:sbp} to the problem of
congressional redistricting in the state of Iowa. Since 2010, Iowa has
been apportioned four congressional districts, and by law, these
districts must be comprised of whole counties, of which there are 99.
Traditionally, maximizing district compactness and minimizing the
deviation in populations across districts are far and away the two most
important criteria for districting plans in Iowa.

We measure district compactness with the \emph{Polsby--Popper} score
\citep{polsby1991third}, which is proportional to the ratio of a
district's area to its perimeter. To score an entire plan, we record the
minimum Polsby--Popper score (least compact) of all the districts: \[
\text{comp}(\xi) \coloneqq \min_{1\le i\le m} 
4\pi \frac{\text{area}_i(\xi)}{\text{perim}_i(\xi)^2},
\] where \(\text{area}_i(\xi)\) and \(\text{perim}_i(\xi)\) denote the
area and perimeter of district \(i\) in a plan \(\xi\). We measure
population equality with the \emph{maximum population deviation} score.
For districts in a plan \(\xi\) with populations \(N_1,\dots,N_m\) in a
state with total population \(N\), the population deviation score is
defined as \[
    \text{dev}(\xi) \coloneqq \max_{1\le i\le m} \abs{\frac{N_i - N/m}{N/m}},
\] that is, the maximum percentage deviation in any district from full
population equality. We can therefore write the overall scoring function
as \(f(\xi) = (-\text{dev}(\xi), \text{comp}(\xi))\).

\begin{figure}[htb]

\centering{

\includegraphics{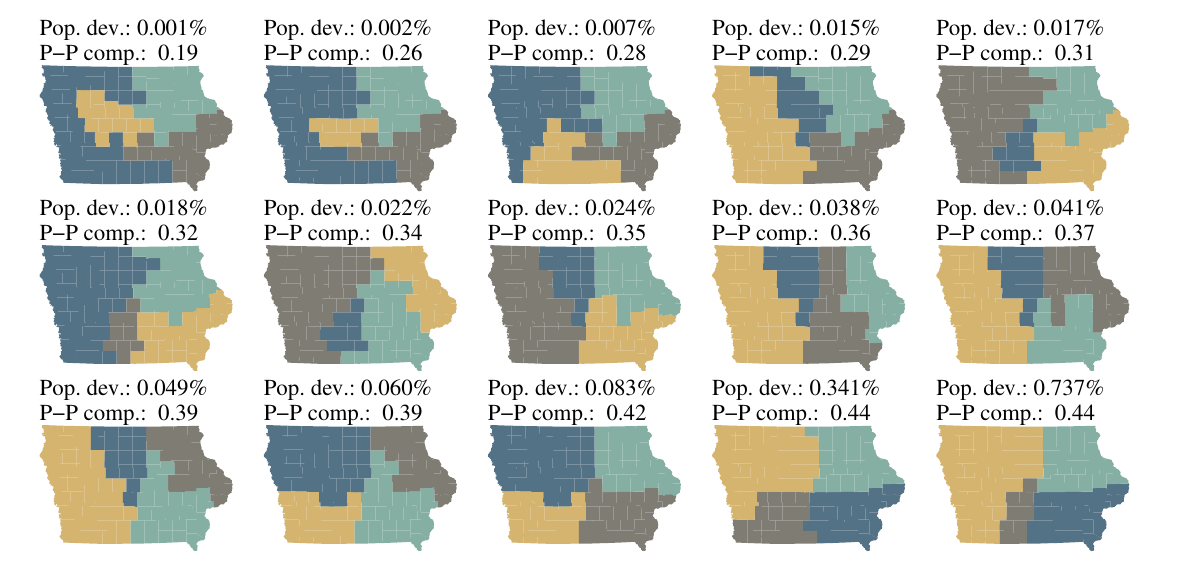}

}

\caption{\label{fig-maps}Plans along the Pareto frontier estimated with
10,000 bursts.}

\end{figure}%

First, we run Algorithm \ref{alg:sbp} with \(b=10\) and \(\xi_0\) set to
the enacted 2020 plan over a range of maximum bursts from 10 to 10,000.
The estimated Pareto frontiers from each of these runs of the algorithm
is plotted in Figure~\ref{fig-demo}. For comparison, we also run the
underlying Markov chain for 1,000 steps, and plot the sampled plans'
compactness and deviation scores on the same figure. Unsurprisingly, the
estimated Pareto frontier lie outside of (or close to) the convex hull
of the sampled plans' scores, even for small numbers of total bursts,

Figure~\ref{fig-maps} plots the 11 plans that define the Pareto frontier
estimated with 10,000 bursts. The minimum-deviation plan is relatively
noncompact compared to the rest of the frontier, though is typical
compared to the set of plans sampled directly from the Markov chain. As
we travel along the frontier, compactness increases substantially,
before reaching a relative plateau. In fact, there appears to be little
overall tradeoff between the two criteria, as indicated by the sharp
angle in the Pareto frontier. Across most of their range, each criterion
can be optimized with minimal effect on the other criterion; only at
extreme values must compactness and population equality be weighed
against each other. Compared to single-criterion optimizers or existing
sampling methods, the proposed algorithm enables qualitative findings
like this which may be useful to redistricting practitioners looking to
improve a districting plan along multiple dimensions.

\begin{figure}

\centering{

\includegraphics{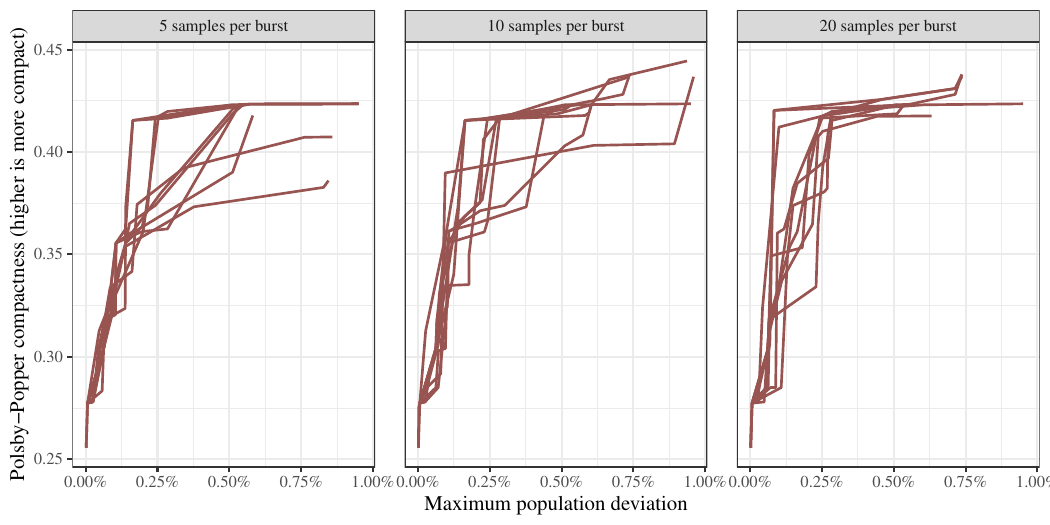}

}

\caption{\label{fig-burst-size}Estimated Pareto frontier across a range
of burst sizes, with 10 replications of 200 bursts each.}

\end{figure}%

Figure~\ref{fig-demo} plots the Pareto frontier as a function of the
total number of bursts. We can also examine the sensitivity of the
estimated frontier to the burst size parameter \(b\). For burst sizes
\(b\in\{5, 10, 20\}\), we run 10 replications of Algorithm \ref{alg:sbp}
for 200 bursts each. The resulting estimated Pareto frontiers are shown
in Figure~\ref{fig-burst-size}.

The frontiers are remarkably consistent across both replications and
varying burst sizes. The lack of sensitivity to the burst size over a
reasonable range of sizes was noted by \citet{cannon2023voting}, and it
is encouraging to see similar results in the multidimensional case. The
consistency across replications further provides confidence that running
Algorithm \ref{alg:sbp} a small-to-moderate number of times will
generally capture a reasonable Pareto frontier (for that computational
budget).

\section{Conclusion}\label{sec-conclude}

We have proposed a natural generalization of the short bursts algorithm
of \citet{cannon2023voting} which will allow practitioners and
researchers to estimate the Pareto frontier induced by a set of criteria
on districting plans. In a demonstration application in the state of
Iowa, we find that the proposed Pareto optimization by short bursts
algorithm is not particularly sensitive to the burst size, and produces
relatively consistent results across multiple independent runs.

Future work should compare the performance of both the univariate and
multi-criteria short burst algorithms to competing redistricting
optimization approaches, including those of \citet{liu2016pear} and
\citet{swamy2022multiobjective}. Additional simulation studies examining
the scaling behavior of all of these algorithms in the number of
districts, number of precincts or counties, and number of redistricting
criteria (dimensionality of \(f\)) would be highly valuable as well.

\hypertarget{refs}{}

\begin{CSLReferences}{0}{0}\end{CSLReferences}

\appendix

\renewcommand\thefigure{\thesection\arabic{figure}}

\setcounter{figure}{0}

\section{Proofs of Propositions}\label{app:proofs}

\propconv*
\begin{proof}
Let $\xi \in P(f)$.
Since every pair of plans has strictly positive transition probability under $\M^b$, with probability 1 the algorithm started at any $A$ will eventually transition to $\xi$ within a burst of length $b$ (i.e., $\exists n : \xi \in X_n$).
Since $\xi\in P(f)$, at the end of the burst, $\xi$ will be added to the approximate Pareto frontier $X$.
Additionally, $\xi$ will never be removed from $X$, since no plan in $\Xi$ Pareto dominates it.
Since this holds for all $\xi\in P(f)$ (a finite set), with probability one all plans in $P(f)$ will eventually belong to $X$.
Then no other plan can belong to $X$, or else such a plan would not be Pareto dominated by any plan in $\Xi$, and thus would belong to $P(f)$ itself.
So $X_n\cvas P(f)$.
\end{proof}

\propirred*
\begin{proof}
Without loss of generality, number the plans of $\Xi$ such that \[ f(\xi_1)=\dots=f(\xi_{i_1})>f(\xi_{i_1+1})=\dots=f(\xi_{i_2})>f(\xi_{i_2+1})\dots f(\xi_{i_{b+1}})<f(\xi_{i_{b+1}+1})=\dots=f(\xi_{i_{b+2}})
\] with $f(\xi_{i_b+1})>f(\xi_1)$.
Then define $\M$ to be a random walk along this ordering; clearly $\M$ is irreducible.
However, initializing Algorithm \ref{alg:sbp} with $A=\{\xi_1\}$, any burst of length $b$ will return a plan $\xi_l$ with $1\le l\le i_{b+1}$, which will be Pareto dominated by one of $\{\xi_1, ... \xi_{i+1}\}$.
Since there are at least $b$ plans separating any of  $\{\xi_1, ... \xi_{i+1}\}$ from the Pareto frontier $P(f)=\{f(\xi_{i_{b+1}+1}),\dots,f(\xi_{i_{b+2}})\}$, no burst of length $b$ will transition into the Pareto frontier, and so $X\not\cvas P(f)$.
\end{proof}

We expect Proposition \ref{prop:irred} to hold for \(J>1\) as well.

\newpage

\section{Scaling of the Pareto frontier
size}\label{scaling-of-the-pareto-frontier-size}

Figure~\ref{fig-frontier-size} shows the scaling behavior of the size of
the Pareto frontier. We generate \(n\) samples from a
\(\mathcal{N}(0, I_{J\times J})\) distribution, for a range of \(n\) and
\(J\).

As both parameters increase, the size of the Pareto frontier grows
rapidly; as \(n\to\infty\) the growth is exponential in \(J\).

\begin{figure}

\centering{

\includegraphics{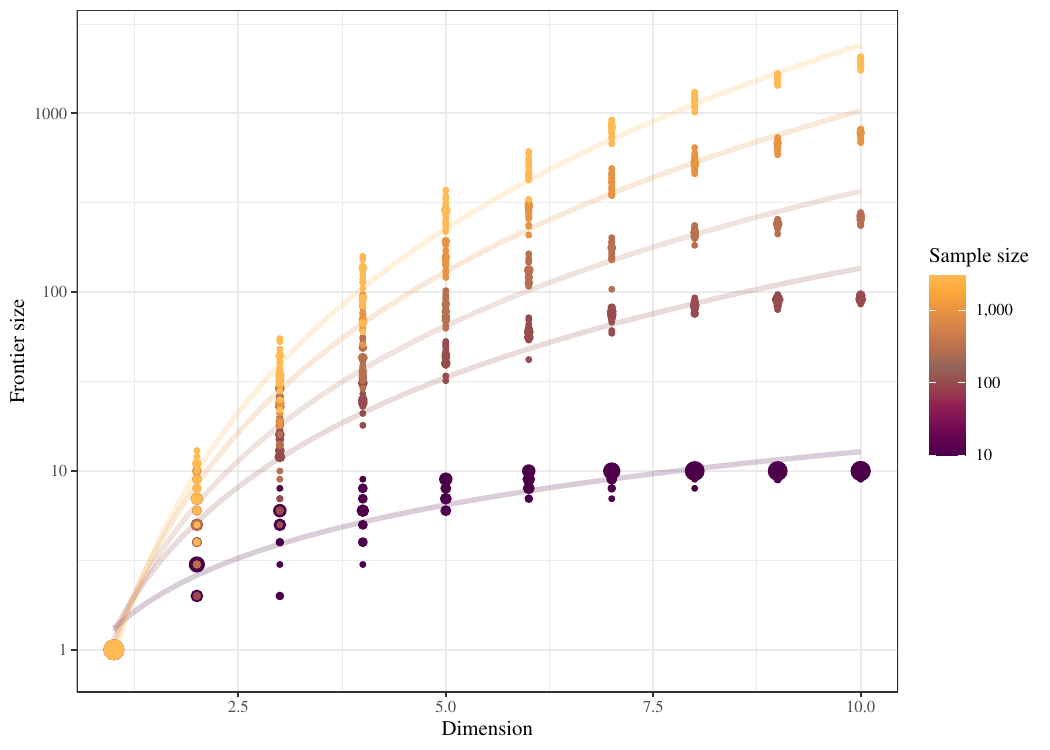}

}

\caption{\label{fig-frontier-size}Pareto frontier size for samples from
a multivariate Normal distribution, by dimension and sample size. Twenty
samples were generated for each combination of parameters.}

\end{figure}%


\end{document}